\documentclass[11pt,a4paper]{article}
\usepackage[utf8]{inputenc}
\usepackage[T1]{fontenc}
\usepackage{amsmath}
\usepackage{amsfonts}
\usepackage{amssymb}
\usepackage{amsthm}
\usepackage[
left=3.4cm, right=3.4cm, top=3.1cm, bottom=3.1cm
]{geometry}
\usepackage{authblk}
\usepackage{textcomp} 
\usepackage{graphics}
\usepackage{bm}
\usepackage{bussproofs} 
\usepackage{stmaryrd} 
\usepackage{xcolor} 
\usepackage{enumitem}
\usepackage{float}
\usepackage{framed} 
\usepackage{multicol}

\allowdisplaybreaks

\newtheorem{theorem}{Theorem}[section]
\newtheorem{lemma}[theorem]{Lemma}
\newtheorem{proposition}[theorem]{Proposition}

\newtheorem{definition}[theorem]{Definition}

\newtheorem{claim}[theorem]{Claim}

\makeatletter  
\def\th@plain{%
  \thm@notefont{}
  \itshape 
}
\def\th@definition{%
  \thm@notefont{}
  \normalfont 
}
\makeatother

\newcommand{\f}{\varphi}
\newcommand{\ff}{\psi}

\newcommand{\To}{\Rightarrow}
\newcommand{\tot}{\leftrightarrow}

\newcommand{\relstruct}[1]{\boldsymbol{#1}} 

\newcommand{\logic}[1]{\bm{\mathsf{#1}}}

\let\rs\relstruct
\let\ax\system

\newcommand{\vp}{\varphi}

  \title{Relevant Reasoners in a Classical World}
  \author[1]{Igor Sedlár}
  \author[2]{Pietro Vigiani}
  \affil[1]{Czech Academy of Sciences, Institute of Computer Science}
  \affil[2]{Scuola Normale Superiore, Department of Philosophy}

\begin{document}

\maketitle

  \begin{abstract}
  We develop a framework for epistemic logic that combines relevant modal logic with classical propositional logic. In our framework the agent is modelled as reasoning in accordance with a relevant modal logic while the propositional fragment of our logics is classical. In order to achieve this feature, we modify the relational semantics for relevant modal logics so that validity in a model is defined as satisfaction throughout a set of designated states that, as far as propositional connectives are concerned, behave like classical possible worlds. The main technical result of the paper is a modular completeness theorem parametrized by the relevant modal logic formalizing the agent's reasoning.
  
  \end{abstract}
 
 \section{Introduction}
 
 The \emph{logical omniscience problem} consists in a discrepancy between properties of modal operators in normal modal epistemic logics on one hand and intuitions concerning epistemic attitudes of real-life agents on the other hand. More specifically, if $\f_1 \land \ldots \land \f_n \to \ff$ is valid in a normal modal logic, for some $n \geq 0$, then so is $\Box \f_1 \land \ldots \land \Box \f_n \to \Box \ff$, meaning informally that epistemic attitudes of moderately idealised agents represented by $\Box$ are closed under local consequence from arbitrary finite (possibly empty) sets of assumptions. As an assumption about  agents with bounded memory, reasoning capacity etc., this is clearly unrealistic.
 
 A number of alternative approaches to epistemic logic that avoid the logical omniscience problem have been proposed. One group of approaches, the \emph{non standard states approaches}, consists in enriching Kripke models with so-called non-standard states in which the behaviour of logical connectives may differ from their behaviour in standard states (possible worlds). Closure of epistemic attitudes under logical consequence is avoided by (i) defining consequence in terms of standard states, but (ii) allowing the epistemic accessibility relation to connect standard states with non-standard ones. One of the first approaches of this kind is \cite{Rantala1982a}; see also \cite{Wansing1990} for a generalization.
 
 In the most general versions of the non-standard states approach, agents' epistemic attitudes are represented as lacking any kind of logical regimentation. It seems more realistic to assume closure under consequence of some weak \emph{non-classical logic} instead of Scylla of logical anarchism and Charybdis of closure under normal modal logic. In this vein, Levesque \cite{Levesque1984} has advocated an epistemic logic based on the logic $\logic{FDE}$ of First Degree Entailment, an implication-free fragment of the relevant logic $\logic{E}$; see also \cite{Lakemeyer1987} for a version allowing nesting of modal operators. In this logic, closure of epistemic attitudes under some problematic principles involving negation is avoided; for instance, agents are allowed to have inconsistent yet non-trivial beliefs. Validity is defined in terms of standard possible worlds and so the propositional fragment of Levesque's logic is classical. The general motivation behind this approach is that Boolean connectives seem to correctly represent the logic of agent-independent facts, while agents' reasoning is regimented by a weaker non-classical logic. This approach will be followed in this paper.
 
 Levesque acknowledges that an extension of his formalism with a non-standard implication connective would be desirable to properly formalize epistemic attitudes towards implicational statements. Fagin et al.~\cite{FaginEtAl1995a} study an extension of $\logic{FDE}$ with epistemic modal operators and material implication, thus providing some means to formalize attitudes towards implicational statements, but the well known discrepancy between properties of material implication and the intuitive properties of implicational statements renders their formalization problematic. 
 
 Relevant modal logic including a relevant implication connective is used as a basis for epistemic logic in \cite{BilkovaEtAl2016,BilkovaEtAl2010}. As is usual in relevant logic, this framework does not take standard states to be classical possible worlds. Hence, it is not an option if one intends to extend classical propositional logic with epistemic modal operators regimented by a relevant logic. Such an option is provided in \cite{Sedlar2015,Sedlar2016a} where combinations of classical propositional logic with relevant modal logic are studied. The semantics of \cite{Sedlar2016a} is two-sorted and the Hilbert-style axiomatizations provided in the paper use a meta-rule of inference, which is inconvenient if the underlying relevant logic is undecidable. The semantics of \cite{Sedlar2015} is one-sorted, thus more elegant, but the Hilbert-style axiomatizations provided in that paper have two peculiar features: proofs are defined non-standardly as \emph{pairs} of finite sequences of formulas, and the completeness proof relies on the presence of extensional truth constants $\top, \bot$ in the language.
 
 In this paper we provide a framework for epistemic logic based on relevant modal logic that avoids the problematic features of the previous approaches. The crucial features of our approach are that (i) closure of epistemic attitudes under normal modal consequence is avoided in all forms (that is, taking into account zero, one or more assumptions) but (ii) epistemic attitudes are represented as being regimented by relevant modal logic while (iii) the non-modal fragments of our epistemic logics remain classical; (iv) our use of full relevant logic, including a relevant implication connective, enables a more realistic representation of attitudes towards implicational statements; and (v) our main technical result, a modular completeness theorem parametrized by the relevant modal logic governing the agents' reasoning, is stated in terms of a standard Hilbert-style proof system, and to obtain the result we do not need to assume the presence of extensional truth constants in the language.
 
 The paper is structured as follows. In Section \ref{sec:relevant_modal} we give some background information on relevant modal logic and in Section \ref{sec:motivation} we outline the general strategy behind our framework. The framework itself is introduced in Section \ref{sec:framework} and the axiomatization results are obtained in Section \ref{sec: axiomatisation}. The concluding section summarizes the paper and discusses some attractive topics for future research. Proofs of some of our results are given in the appendix.

\section{Relevant modal logics}\label{sec:relevant_modal}

This section gives some background information on relevant modal logic. We build on the approach of Fuhrmann \cite{Fuhrmann1988,Fuhrmann1990}, with the difference that relevant modal logics discussed here are \textit{bimodal}, with a pair of modal operators $\Box$ and $\Box_L$. While $\Box$ will be seen as an epistemic modal operator, the auxiliary operator $\Box_L$ is introduced to capture ``validity in relevant logic'' in a specific technical sense that will be discussed later.

\begin{definition}
The \emph{modal language $\mathcal{L}$} contains a countable set $Pr$ of propositional variables and operators $\land, \lor, \to$ (binary), and $\neg, \Box, \Box_L$ (unary). The set of formulas of this language is defined as usual and denoted as $Fm_{\mathcal{L}}$. We define $\f \tot \ff := (\f \to \ff) \land (\ff \to \f)$.
\end{definition}
The propositional operators $\land, \lor, \neg$ and $\to$ are read as usual. The modal operator $\Box$ is read (with a contextually fixed agent in mind) as ``the agent believes that ...'', while the operator $\Box_L$ is auxiliary. We note that we stick to a language with one epistemic modal operator $\Box$, instead of multiple $\Box_i$ for agents $i \in G$, only for the sake of simplicity and that the extension of our framework to the multi-agent case is trivial.

Before defining the semantics for the modal language, we introduce some simplifying notation. Let $(S_1, \leq_1)$ and $(S_2, \leq_2)$ be two partially ordered sets. If $k_1, \ldots, k_n, k_{n+1} \subseteq \{ \downarrow, \uparrow \}$, then an $n$-ary function $f$ from $(S_1, \leq_1)$ to $(S_2, \leq_2)$ is said to be of type $k_1\ldots k_n \mapsto k_{n+1}$ iff
$$\bigwedge_{i \leq n} \big( s_i Z_i t_i \big) \implies  f(s_1, \ldots, s_n) Z_{n+1} f(t_1, \ldots, t_n) $$ where $Z_i = \mathord{\leq}$ in case $k_i = \mathord{\uparrow}$ and $Z_i = \mathord{\geq}$ in case $k_i = \mathord{\downarrow}$. We denote as $S_1(k_1 \ldots k_n, S_2(k_{n+1}))$ the set of $n$-ary functions from $S_1$ to $S_2$ of type $k_1 \ldots k_n \mapsto k_{n+1}$. As a special case, $n$-ary relations on $(S, \leq)$ are $n$-ary operations from $(S, \leq)$ to $T = (\{ \textit{true}, \textit{false} \}, \sqsubseteq )$, where it is assumed that $\textit{false} \sqsubseteq \textit{true}$. For example, $S(\uparrow, T(\uparrow))$ denotes the set of all subsets of $S$ that are closed upwards under $\leq$; $S(\downarrow\uparrow, T(\uparrow))$ denotes the set of binary relations on $S$ that are anti-monotonic in the first position and monotonic in the second position; and $S(\uparrow, S(\downarrow))$ denotes the set of anti-monotonic unary functions on $S$. We will usually omit $T(\uparrow)$; hence $S(\uparrow)$ means $S(\uparrow, T(\uparrow))$. If $B$ is a binary relation on $S$, then $B(s)$ denotes the set $\{ t \mid Bst \}$, and if $X \subseteq S$, then $B(X) := \bigcup_{s \in X} B(s)$.

\begin{definition}
A \emph{frame} is a tuple $F = (S, \leq, R, *, Q, Q_L)$ where $(S, \leq)$ is a partially ordered set, $R \in S(\downarrow\downarrow\uparrow)$, $*  \in S(\uparrow, S(\downarrow))$ and $Q, Q_L \in S(\downarrow\uparrow)$. A \emph{model} based on a frame $F$ is $M = (F, V)$ where $V: Pr \to S(\uparrow)$.
\end{definition}

We will consider later on structures that expand frames by additional relations on $S$ and we will apply the terminology defined for frames to these structures.

\begin{definition}
For each frame $F$, we define the following operations on $2^{S}$:

\begin{gather*}
    X \land^{F} Y = X \cap Y\\
    X \lor^{F} Y = X \cup Y\\
    \neg^{F} X = \{ s \mid s^{*} \not\in X \}\\
    X \circ^{F} Y = \{ u \mid \exists s,t (s \in X \And t \in Y \And Rstu) \}\\
    X \to^{F} Y = \{ s \mid \{ s \} \circ^{F} X \subseteq Y \}\\
    \Box^{F} X = \{ s \mid \forall t (Qst \To t \in X) \}\\
    \Box_L^{F} X = \{ s \mid \forall t (Q_Lst \To t \in X) \}
\end{gather*}
\end{definition}
These operations are related to the standard satisfaction clauses of relevant (modal) logic. The tonicity conditions incorporated into the definition of a frame ensure that $S(\uparrow)$ is closed under $c^{F}$ for $c \in \{ \land, \lor, \to, \neg, \Box, \Box_L \}$.

\begin{definition}\label{def:interpretation}
For each $M = (F, V)$, the \emph{$M$-interpretation} is a function $\llbracket \, \rrbracket_{M} : Fm_{\mathcal{L}} \to S(\uparrow)$ such that $\llbracket p\rrbracket_M = V(p)$ and
\begin{equation*}
\llbracket c(\f_1, \ldots, \f_n)\rrbracket_{M} = c^{F} \big ( \llbracket \f_1\rrbracket_{M}, \ldots, \llbracket \f_n\rrbracket_{M} \big)
\end{equation*}
for all $c \in \{ \land, \lor, \to, \neg, \Box, \Box_L \}$.
\end{definition}
 Since $S(\uparrow)$ is closed under each $c^{F}$ and $V(p) \in S(\uparrow)$ by definition, it follows that $\llbracket \f\rrbracket_M \in S(\uparrow)$ for all $\f \in Fm_{\mathcal{L}}$. We will often write $(M, s) \models \f$ instead of $s \in \llbracket\f\rrbracket_{M}$ and $(M, X) \models \Gamma$ instead of $(M, s) \models \f$ for all $\f \in \Gamma$ and $s \in X$ ($s \models \f$, $X \models \Gamma$ when $M$ is clear from context); and we will not distinguish between singletons and their elements when using this notation. 
 
 As usual in relevant logic, elements of $S$ are seen as bodies of information, or \emph{situations}, roughly in the sense of \cite{BarwisePerry1983}, partially ordered by the amount of information they support. Situations are not closed under the usual laws of classical logic, and the relations used to define the operations corresponding to $\neg$ and $\to$ are introduced to achieve this. First, situations may be incomplete or inconsistent, i.e. for some $s \in S$ we may have $s \in \llbracket \f\rrbracket_M \cap \llbracket \neg\f\rrbracket_M$ (meaning that $s \not\leq s^*$) or $s \not\in \llbracket \f \rrbracket_M \cup \llbracket \neg \f\rrbracket_M$ (meaning that $s^* \not\leq s$). Informally, $s^{*}$ is seen as the maximal situation that is \emph{compatible} with $s$. Second, note that $s \in X \to Y$ iff, for all $t,u$, if $Rstu$ and $t \in X$, then $u \in Y$. Interpretation of implication in terms of a ternary relation on situations enables the failure of the so-called ``paradoxes of strict implication''; note that it may be the case that $s \in \llbracket \f\rrbracket_M$ and $s \not\in \llbracket \f \to \f \rrbracket_M$ if there are $t,u$ such that $Rstu$ and $t \not\leq u$. Informally, the ternary relation $R$ is seen as representing \emph{combination} of information supported by situations; $Rstu$ can be seen as representing the fact that the body of information that results from combining the information supported by $s$ with the information supported by $t$ is contained in the information supported by $u$.\footnote{For instance, Dunn and Restall point out that ``perhaps the best reading [of $Rstu$] is to say that the combination of the pieces of information $s$ and $t$ (not necessarily the union) is a piece of information in $u$'' \cite[p.\ 67]{DunnRestall2002}.} Not much is assumed about $R$ in the general setting; for instance, we do not assume that information combination is commutative ($Rstu \To Rtsu$), associative ($Rstuv \To Rs(tu)v$; see the explanation of the notation before Figure \ref{fig:frame_conditions}) or reflexive ($Rsss$). However, these and similar properties of $R$ can be assumed when one considers stronger relevant logics; see Figure \ref{fig:frame_conditions}. The modal accessibility relation $Q$ represents information about the beliefs of our contextually fixed agent. In particular, $Qst$ expresses that all information that the agent believes according to $s$ is supported by $t$. 
 
 We now define an expansion of frames that is used in relational semantics for relevant logics. The key feature is a semantic deduction theorem according to which $\f \to \ff$ is valid iff $\ff$ follows from $\f$ in the frame.
 
 \begin{definition}
 An \emph{$L$-frame} is a structure $\rs{F} = (F, L)$, where $F$ is a frame and $L \in S(\uparrow)$ such that
 \begin{gather}
 \forall s \exists x (x \in L \And Rxss) \label{eq:L1}\\
 s \in L \And Rstu  \To t \leq u \label{eq:L2}
\end{gather}
A \emph{model} based on an $L$-frame $\rs{F}$ (an $L$-model) is a tuple $\rs{M} = (\rs{F}, V)$ where $V: Pr \to S(\uparrow)$. A formula $\f$ is \emph{valid} in an $L$-model iff $L \subseteq \llbracket\f\rrbracket_{\rs{M}}$; notation $M \models \f$. A formula $\f$ is valid in a class of $L$-frames iff it is valid in each $L$-model based on an $L$-frame in the class.
 \end{definition}
 
 \begin{lemma}\label{lem:ded_L-models}
 For all $L$-models $\rs{M}$, $\f \to \ff$ is valid in $\rs{M}$ iff $\llbracket \f\rrbracket_{\rs{M}} \subseteq \llbracket \ff\rrbracket_{\rs{M}}$.
 \end{lemma}
 \begin{proof}
 Frame conditions \eqref{eq:L1}, \eqref{eq:L2} and the fact that $\llbracket \f\rrbracket_{\rs{M}} \in S(\uparrow)$ for all $\f$.
 \end{proof}
 
 The set of formulas valid in all $L$-models is denoted as $\logic{BM.C}$. We use the notation of \cite{Fuhrmann1988,Fuhrmann1990}, where $\logic{BM.C}$ denotes the smallest \textit{conjunctively regular} modal extension of $\logic{BM}$. A modal logic is said to be conjunctively regular if its modal operators distribute over conjunctions. The logic $\logic{BM}$ is one of the weakest propositional relevant logics; see \cite[I.3--I.4]{Fuhrmann1988}.
 
 \begin{definition}
The \emph{axiom system} $\ax{BM.C}$ consists of the following axioms
\begin{align*}
    &\text{(a1)}& &p \to p& &\text{(a7)}& & q \to (p \lor q) &\\
    &\text{(a2)}& & \neg (p \land q) \to (\neg p \lor \neg q) & &\text{(a8)}& & ((p \to q) \land (p \to r)) \to (p \to (q \land r)) & \\
    &\text{(a3)}& &(\neg p \land \neg q) \to \neg (p \lor q) & &\text{(a9)}& & ((p \to r) \land (q \to r)) \to ((p \lor q) \to r) &\\
    &\text{(a4)}& & (p \land q) \to p & &\text{(a10)}& & (p \land (q \lor r)) \to ((p \land q) \lor (p \land r)) &  \\
    &\text{(a5)}& & (p \land q) \to q &  &\text{(a11)}& &(\Box p \land \Box q) \to \Box (p \land q)&\\
    &\text{(a6)}& & p \to (p \lor q) & &\text{(a12)}& &(\Box_L p \land \Box_L q) \to \Box_L (p \land q) &
    \end{align*}
plus the rules of Uniform substitution (US) and Modus ponens (MP) and
\begin{center}
    $\text{\LeftLabel{(Adj)}\AxiomC{$ \vp$}\AxiomC{$ \psi$}\BinaryInfC{$\vp \land \psi$}\DisplayProof}$
    \quad
    $\text{\LeftLabel{(Aff)}\AxiomC{$ \vp' \to \vp$}\AxiomC{$ \psi \to \psi'$}\BinaryInfC{$(\vp \to \psi) \to (\vp' \to \psi')$}\DisplayProof}$
    \quad
    $\text{\LeftLabel{(Con)}\AxiomC{$ \vp \to \psi$}\UnaryInfC{$\neg \psi \to \neg\vp$}\DisplayProof}$\\[2mm]
    $\text{\LeftLabel{($\Box_L$-Mon)}\AxiomC{$ \vp \to \psi$}\UnaryInfC{$\Box_L \vp \to \Box_L \psi $}\DisplayProof}$
    \quad
    $\text{\LeftLabel{($\Box$-Mon)}\AxiomC{$ \vp \to \psi$}\UnaryInfC{$\Box \vp \to \Box \psi $}\DisplayProof}$
    \end{center}    

\end{definition}
The set of theorems of $\mathsf{BM.C}$, $Th(\mathsf{BM.C})$, is defined in the usual way.
The axiom system $\ax{BM.C}$ given here differs from Fuhrmann's axiom system for the conjunctively regular modal extension of $\logic{BM}$ as follows: instead of axiom schemata we use axioms and add (US); we add a conjunctive regularity axiom and a monotonicity rule for the second modal box operator $\Box_L$. Note that we do not assume any interplay between $\Box$ and $\Box_L$.

Some frame conditions assumed in the semantics for relevant modal logics stronger than $\logic{BM.C}$ are listed in Figure \ref{fig:frame_conditions}, where we define $Rstuv := \exists x (Rstx \And Rsuv)$, $Rs(tu)v := \exists x (Rsxv \And Rtux)$, $RQstu := \exists x (Rstx \And Qxu)$ and $QRstu := \exists x (Qsx \And Rxtu)$. 

\begin{figure}
\centering
\begingroup
\setlength{\tabcolsep}{2mm}
\begin{tabular}{l l l}
\multicolumn{2}{l}{\textbf{Frame condition}}\qquad &\textbf{Corresponding axiom/rule} \\[-2mm] \multicolumn{3}{c}{\hrulefill}\\ 
(DN) & $s^{**} = s$ & $p \tot \neg\neg p$\\
(Cp) & $Rstu \To Rsu^{*}t^{*}$ & $(p \to q) \to (\neg q \to \neg p)$\\
(WB) & $Rstu \To Rs(st)u$ & $((p \to q) \land (q \to r)) \to (p \to r)$\\
(X) & $s \in L \To s^{*} \leq s$ & $p \lor \neg p$\\
(Rd) & $Rss^{*}s$ & $(p \to \neg p) \to \neg p$\\
(B) & $Rstuv \To Rs(tu)v$ & $(p \to q)  \to ((r \to p) \to (r \to q))$\\
(CB) & $Rstuv \To Rt(su)v$ & $(p \to q) \to ((q \to r) \to (p \to r))$\\
(W) & $Rstu \To Rsttu$ & $(p \to (p \to q)) \to (p \to q)$\\
(C) & $Rstuv \To Rsutv$ & $(p \to (q \to r)) \to (q \to (p \to r))$\\
(M) & $Rstu \To (s \leq u \lor t \leq u)$ & $p \to (p \to p)$\\
(ER) & $\exists x (x \in L \And Rsxs)$ & $\dfrac{\f}{(\f \to \ff) \to \ff}$ \\
(Nec) & $(x \in L \And Q xs) \To s \in L$ & $\dfrac{\f}{\Box \f}$\\
($\Box$K) & $RQstu \To \exists x (Qtx \And QRsxu)$ & $\Box (p \to q) \to (\Box p \to \Box q)$\\
($\Box$T) & $Qss$ & $\Box p \to p$\\
($\Box$D) & $\exists x (Qs x^{*} \And Qs^{*}x)$ & $\Box \neg p \to \neg \Box p$\\
($\Box$4) & $(Qst \And Qtu) \To Qsu$ & $\Box p \to \Box \Box p$\\
($\Box$5) & $(Qs^{*}u \And Qst) \To Qt^{*}u$ & $\neg \Box p \to \Box \neg \Box p$ \\
\multicolumn{3}{c}{\hrulefill}
\end{tabular}\caption{Some prominent frame conditions with the corresponding axioms and rules.}\label{fig:frame_conditions}
\endgroup
\end{figure}

For example, the logic $\logic{DW.C}$ is defined in terms of frames satisfying (Cp), the logic $\logic{TW.C}$ adds (B) and (CB), $\logic{T.C}$ adds (WB), (X), (RD) and (W), $\logic{E.C}$ adds the rule (ER), $\logic{R.C}$ adds (C) and the logic $\logic{RM.C}$ adds (M). For a propositional relevant logic $\logic{PL}$, the logic $\logic{PL.R}$ adds \textit{implicational regularity} ($\Box$K) to $\logic{PL.C}$ and $\logic{PL.K}$ further adds (Nec). If $\logic{X} \in \{ \logic{C}, \logic{R}, \logic{K} \}$, then $\logic{PL.XT}$ adds ($\Box$T) to $\logic{PL.X}$, and similarly for $\logic{D}$, $\logic{4}$ and $\logic{5}$. For a more extensive list, including especially more variation on the propositional level, see \cite{Fuhrmann1988,Fuhrmann1990,RoutleyEtAl1982}.

If a relevant modal logic $\logic{L}$ is defined as a set of formulas valid in all $L$-frames satisfying a selection of frame conditions from Figure \ref{fig:frame_conditions}, then we will denote the set of $L$-frames satisfying those frame conditions as \emph{$L$-frames for $\logic{L}$}, or simply $\logic{L}$-frames, and the selection of the frame conditions as \emph{$\logic{L}$-conditions}. For any relevant modal logic $\logic{L}$, the axiom system $\ax{L}$ is obtained by adding to $\ax{BM.C}$ the axioms and rules corresponding to the $\logic{L}$-conditions according to Figure \ref{fig:frame_conditions}.

\begin{theorem}\label{thm:completeness-L-models}
For all $\logic{L}$, $\logic{L} = Th(\ax{L})$.
\end{theorem}
\begin{proof}
Virtually the same argument as the one in \cite{Fuhrmann1988,Fuhrmann1990}. 
\end{proof}

\section{Classical epistemic logic based on relevant modal logic}\label{sec:motivation}

It is clear that relevant modal logics, as presented above, can be used to model reasoning about agents that are not logically omniscient with respect to classical logic, but whose beliefs are still logically regimented. For instance, $\Box (p \land \neg p) \to \Box q$ is invalid in all $\logic{L}$  mentioned above, but $\Box (p \land q) \to \Box p$ is valid. Interestingly, the degree to which agents are represented as being omniscient with respect to the relevant modal logic at hand is smaller than in the case of classical modal logic. As the Regularity axiom (a11) and the Monotonicity rule ($\Box$-Mon) entail, each $\mathsf{L}$ is closed under
\begin{equation}\label{eq:MR}
\dfrac{\f_1 \land \ldots \land \f_n \to \ff}{\Box\f_1 \land \ldots \land \Box \f_n \to \Box \ff} \tag{CR}
\end{equation}
for $n > 0$, but not each $\logic{L}$ is closed under the Necessitation rule (Nec). This means that even though the beliefs of agents are assumed to be closed under valid implications with non-empty antecedents, agents are not assumed to believe all valid formulas. 

Moreover, in general $\logic{L}$ is not closed under the implicational version of \eqref{eq:MR}
\begin{equation}\label{eq_MR-imp}
\dfrac{\f_1 \to (\ldots (\f_n \to \ff) \ldots)}{\Box \f_1 \to (\ldots (\Box \f_n \to \Box \ff) \ldots)} \, . \tag{IR}
\end{equation}
Relational semantics for relevant modal logics yield conjunctively regular but not necessarily implicatively regular modal logics. For instance, even in logics with the implicational version of contraposition (Cp) as an axiom, the formula $\Box (p \to q) \to (\Box \neg q \to \Box \neg p)$ is not necesarily valid. In this respect relevant modal logic is more fine-grained than classical modal logic, where the distinction between \eqref{eq:MR} and \eqref{eq_MR-imp} collapses due to the provability of $(\vp \to (\psi \to \chi)) \tot ((\vp \land \psi) \to \chi)$\footnote{This will be the case also for any $\mathsf{L}$ containing the K-axiom ($\Box$K), as well as for  classical logic, where $(\vp \to (\psi \to \chi)) \tot ((\vp \land \psi) \to \chi)$ is  provable.}.
The difference  between \eqref{eq:MR} and \eqref{eq_MR-imp} expresses the assumption that while beliefs of agents are represented as ``automatically'' closed under conjunction introduction, they are not seen as closed under implication elimination. As Sequoiah-Grayson \cite{SequoiahGrayson2021} points out, this can be understood as meaning that while agents are assumed to automatically \emph{aggregate} their beliefs, they are not assumed to automatically \emph{combine} them. While  this view constitutes a prima facie motivation for \eqref{eq:MR}, belief aggregation has been criticised in the epistemic logic literature, e.g. by \cite{FaginHalpern1987}, who argue that  beliefs tend to come in non-interacting clusters, or frames of mind. Accommodating this idea in our framework leads to a shift from relational to neighborhood semantics, and is left for further work.

The informational interpretations of relevant logics make them good candidates for logics that regiment epistemic attitudes of realistic, non-omniscient agents. One may wonder if such an employment of relevant logics is in conflict with using classical propositional logic. In fact, most frameworks dealing with logical omniscience build on classical propositional logic and add various extra requirements to normal modal logic on the epistemic level; see \cite{BertoHawke2018,FaginHalpern1987,Levesque1984} for instance. It can be argued that while classical consequence models truth preservation, relevant consequence models preservation of informational support in situations and their combinations. The latter is naturally associated with models of how agents reason.

In the rest of this paper we will develop a framework for epistemic logic based on these considerations.  In our framework, the agent is modelled as a \emph{relevant reasoner in a classical world:} the agent reasons in accordance with a relevant modal logic, but the propositional fragment is classical. Our framework extends the framework of \cite{Lakemeyer1987,Levesque1984} by including relevant implication and allowing for a greater variability on the propositional level. 

More specifically, for each relevant modal logic $\mathsf{L}$ we develop a ``classical'' modal logic $\mathsf{CL}$ with two main features. First, the propositional fragment of each $\mathsf{CL}$ is classical propositional logic. Second, the set of theorems of $\mathsf{CL}$ is not closed under (\ref{eq:MR}), but only under the so-called \emph{Relevant reasoning meta-rule}
\begin{equation}\label{eq:RR}
\dfrac{\vdash_{\mathsf{L}} \f_1 \land \ldots \land \f_n \to \ff}{\vdash_{\mathsf{CL}} \Box \f_1 \land \ldots \land \Box \f_n \to \Box \ff} \tag{RR}
\end{equation}
for $n > 0$. In order to ensure that the propositional fragment of $\mathsf{CL}$ is classical propositional logic, $\mathsf{CPC}$, we modify  the relational semantics for relevant modal logics presented above so that validity in a model is defined as satisfaction throughout a set of designated states that, as far as propositional connectives are concerned, behave like classical possible worlds.

The auxiliary modal operator $\Box_L$ is crucial in achieving the second feature of our framework, closure under (\ref{eq:RR}). In particular, we will show that 
\begin{equation*}
\vdash_{\mathsf{L}} \f \to \ff \:\implies\: \vdash_{\mathsf{CL}} \Box_L (\f \to \ff) \:\implies\: \vdash_{\mathsf{CL}} \f \to \ff
\end{equation*}
In fact, we will show that\: $\vdash_{\mathsf{L}} \f \iff \vdash_{\mathsf{CL}} \Box_L \f$. In this sense, $\Box_L$ captures validity in $\mathsf{L}$.

Closure under \eqref{eq:RR} is the key feature of our framework.
We stress that while it is satisfied, the standard logical omniscience problem is avoided in our framework since $\mathsf{CL}$ is generally \emph{not} closed under \eqref{eq:MR} nor under (Nec). This follows from the fact that while validity is defined as satisfaction in all standard states (in our case, possible worlds), the epistemic accessibility relation $Q$ may connect standard states with non-standard states.

\section{Possible worlds in relevant models}\label{sec:framework}

In this section we introduce the central semantic framework of the paper. Our semantics is a modification of the relational semantics for relevant modal logics introduced in Section \ref{sec:relevant_modal} where the set of logical states $L$ is replaced with a set of states $W$ that, as far as propositional connectives are concerned, behave like classical possible worlds. For technical reasons, our frames need to be \emph{bounded} in a specific sense. Bounded models for relevant modal logics were introduced by Seki \cite{Seki2003a,Seki2003}.

\begin{definition}
A \emph{bounded} frame is a frame $F$ where $(S, \leq)$ is a bounded poset, i.e.~there are elements $0, 1 \in S$ such that for all $s \in S$ $0 \leq s \leq 1$, and where, for all $s,t \in S$, the following are satisfied ($Q_{(L)} \in \{ Q, Q_L \}$):
\begin{gather}
1^{*}  = 0 \text{ and } 0^{*} = 1 \label{eq:1*}\\
Q_{(L)} 00 \label{eq:Q_L2}\\
Q_{(L)} 1s  \To s = 1 \label{eq:Q_L3}\\
R010 \label{eq:010}\\
R1st \To (s = 0 \text{ or } t = 1)\label{eq:R1st}
\end{gather}
A \emph{bounded model} is a model $M = (F, V)$ where $F$ is a bounded frame and, for all $p \in Pr$, $0 \not\in V(p)$ and $1 \in V(p)$.
\end{definition}

\begin{lemma}\label{lem:0empty1full}
If $M$ is a bounded model, then
\begin{enumerate}
\item $(M, 1) \models \f$ for all $\f$;
\item $(M, 0) \not\models \f$ for all $\f$.
\end{enumerate}
\end{lemma}
\begin{proof}
By simultaneous induction on the complexity of $\vp$. The base case holds by definition of $V$ in bounded models. The cases of $\vp:= \psi \land \chi$, $\vp := \psi \lor \chi$ are trivial. When $\vp := \neg \psi$, $1 \models \neg \psi$ iff $1^\ast \not\models \psi$ iff by \eqref{eq:1*} $0 \not\models \psi$, which holds by the induction hypothesis (IH); $0 \not\models  \neg\psi$ iff $0^\ast  \models \psi$ iff by \eqref{eq:1*} $1 \models \psi$, which holds by IH.
When $\vp := \psi \to \chi$, $1 \models \psi \to \chi$ iff $\forall s,t \in S (R1st, \ s\models \psi \Rightarrow t\models \chi)$. If $R1st$, then by \eqref{eq:R1st} we distinguish two cases: if $s=0$ then by IH $0 \not\models \psi$, and if $t=1$, then $t \models \chi$ by IH. In both cases $1 \models \f \to \ff$. $0 \not\models \psi \to \chi$ iff $\exists s,t \in S$ such that $R0st, \ s \models \psi$ and $t \not \models \chi$; by \eqref{eq:010} in particular $R010$, so  $1\models \psi$ and $0\not\models \chi$ by IH.
When $\vp:= \Box \psi$, $1 \models \Box \vp$ iff $\forall s \in S$, $Q1s \Rightarrow s \models \psi$; by \eqref{eq:Q_L3}, $Q1s$ entails $s=1$ and by IH $1 \models \psi$. $0 \not \models \Box \psi$ iff $\exists s \in S$ such that $Q0s$ and $s \not\models \psi$; by \eqref{eq:Q_L2} $Q00$ and $0 \not\models \psi$ by IH. The case $\vp:= \Box_L \psi$ is analogous. 
\end{proof}

\begin{definition}
Let $M$ be a bounded model. An element $w \in S$ is a \emph{possible world} iff, for all $s,t \in S$:
\begin{gather}
w^{*} = w\label{eq:world1}\\
Rwww\label{eq:world2}\\
Rwst \To (s = 0 \text{ or } w \leq t)\label{eq:world3}\\ 
Rwst \To (t = 1 \text{ or } s \leq w^{*}) \label{eq:world4}
\end{gather}
\end{definition}

\begin{definition}\label{def:W-frame}
A \emph{$W$-frame} is a structure $\rs{F} = (F, W)$ where $F$ is a bounded frame, $W \subseteq S$ is a set of possible worlds and the following conditions are satisfied:
\begin{gather}
(\forall w \in W)(\forall s,t,u)(Q_L wu \And Rust \To s \leq t) \label{eq:Q_L-s<t}\\
(\forall s)(\exists w \in W)(\exists u)(Q_L wu \And Russ) \label{eq:Q_L-ss}
\end{gather}
A \emph{$W$-model} based on $\rs{F}$ is $\rs{M} = (\rs{F}, V)$ where $V: Pr \to S(\uparrow)$ such that $1 \in V(p)$ for all $p$ and $0 \not\in V(p)$ for all $p \in Pr$.
\end{definition}

Conditions \eqref{eq:Q_L-s<t}-\eqref{eq:Q_L-ss} enable $W$-frames to simulate validity in $L$-models. In $W$-frames, the set of states $Q_L(W) = \{ u \mid \exists w (w \in W  \And Q_Lwu)\}$ ``plays the role'' of $L$ in $L$-models: comparing \eqref{eq:L1}-\eqref{eq:L2} with \eqref{eq:Q_L-s<t}-\eqref{eq:Q_L-ss}, we observe that the set of states $Q_L$-accessible from $W$ has the crucial properties of $L$ in $L$-models. In fact, if $(F, W)$ is a $W$-frame, then $(F, Q_L(W))$ is an $L$-frame. Moreover, we note that we have to explicitly mention the bounds $0,1$ in (\ref{eq:world3}) and (\ref{eq:world4}) since $Rwst \To w \leq t$ and $Rwst \To s \leq w^{*}$ do not hold in the canonical model (see  Section \ref{sec: axiomatisation} and Footnote \ref{ft:01canoncial}).

\begin{definition}
For each $W$-model $\rs{M}$, the \emph{$\rs{M}$- interpretation} $\llbracket\,\rrbracket_{\rs{M}}$ is defined as in Definition \ref{def:interpretation}. A formula $\f$ is \emph{valid} in a $W$-model $\rs{M}$ iff $W \subseteq \llbracket\f\rrbracket_{\rs{M}}$. A formula $\f$ is valid in a class of $W$-frames iff it is valid in each $W$-model based on a $W$-frame belonging to the class.
\end{definition}

\begin{lemma}\label{lem:ded_BoxL}
$\Box_L ( \f \to \ff)$ is valid in a $W$-model $\rs{M}$ iff $\llbracket \f\rrbracket \subseteq \llbracket \ff\rrbracket$.
\end{lemma}
\begin{proof}
\eqref{eq:Q_L-s<t} and \eqref{eq:Q_L-ss}, together with the fact that $\llbracket \f\rrbracket_{\rs{M}} \in S(\uparrow)$, for all $\f$.
%
\end{proof}

\begin{proposition}\label{prop:classconn}
Let $\rs{M}$ be any W-model and $w$ a possible world. Then:
\begin{enumerate}
\item $(\rs{M}, w) \models \neg\f$ iff $(\relstruct{M}, w) \not\models \f$
\item $(\relstruct{M}, w) \models \f \to \ff$ iff $(\relstruct{M}, w) \not\models \f$ or $(\relstruct{M}, w) \models \ff$. 
\end{enumerate}
\end{proposition}
\begin{proof}
The first claim follows easily from $w = w^{*}$. The left-to-right implication of the second claim follows from $Rwww$. The right-to-left implication is established as follows. If $w \models \ff$, then $w \models \f \to \ff$ since $Rwst$ implies $w \leq t$. If $w \not\models \f$ and $Rwst$ with $ s \models \f$, then we reason as follows. If $t = 1$, then $t \models \ff$ by Lemma \ref{lem:0empty1full}. If $t \neq 1$, then $s \leq w = w^{*}$ by \eqref{eq:world4} and so $w \models \f$, which is a contradiction.  
\end{proof}

Note that even though $\neg$ and $\to$ behave like Boolean negation and material implication, respectively, when evaluated in possible worlds, their semantic interpretation is uniform across the model, namely, it is given by the semantic operations $\neg^{\bm{F}}$ and $\to^{\bm{F}}$. The difference in their behaviour is given by the specific properties of the states of evaluation. Hence, we do not assume that the meaning of the symbols $\neg$ and $\to$ is context-dependent in our setting.\footnote{We thank Peter Verdée and Pierre Saint-Germier for pushing us on this point.} 

\begin{definition}
A \emph{C-variant} of an $L$-frame condition $\Phi$ from Figure \ref{fig:frame_conditions} is a first-order formula that results from $\Phi$ by replacing each occurrence of $s \in L$ by $\exists w (w \in W \And Q_L ws)$, where $w$ is a fresh variable.

\emph{$\logic{CL}$ frame conditions} are the C-variants of the $\logic{L}$ frame conditions, for all $\logic{L}$. A $\logic{CL}$-frame is any $W$-frame that satisfies the $\logic{CL}$ frame conditions. A $\logic{CL}$-model is a model based on a $\logic{CL}$-frame.
\end{definition}

It follows from the definition that if $L$ does not occur in frame condition $\Phi$, then $\Phi$ is identical to its C-variant.

We have already noted that each $W$-model can be seen as an $L$-model where $L = Q_L(W)$. Conversely, each $L$-model can be transformed into an ``equivalent'' $W$-model satisfying the ``right'' frame conditions. This fact will be used in the completeness proof in the next section.

\begin{proposition}\label{prop:+}
For each $L$-model $\rs{M}$ for $\logic{L}$ with a set of states $S$ there is a $W$-model $\rs{M'}$ for $\logic{CL}$ with the set of states $S' \supsetneq S$ such that, for all $\f \in Fm_{\mathcal{L}}$:
\begin{enumerate}
\item for all $s \in S$, $(\rs{M}, s) \models \f$ iff $(\rs{M'}, s) \models \f$;
\item if $\rs{M} \not\models \f$, then $\rs{M'} \not\models \Box_L \f$.
\end{enumerate}
\end{proposition}
\begin{proof}
See the appendix.
\end{proof}

We note that, in comparison with \cite{Fuhrmann1988,Fuhrmann1990}, we do not consider logics arising by using two specific frame conditions, namely, the weakening frame condition (K) $Rstu \To s \leq u$ and (M3) $s \in L \And t \leq u \To t \leq s$. The reason for avoiding these is that (i) both are rather strong from the relevant logic perspective and (ii) including them would force us to use a substantially more complicated version of the construction used in the proof of Proposition \ref{prop:+}.

\section{Axiomatization}\label{sec: axiomatisation}

In this section we establish the main technical result of the paper, namely, a modular axiomatization result for logics of the form $\logic{CL}$.

\begin{definition}
Let $\mathsf{L}$ be the axiom system for one of the relevant modal logics discussed in Section \ref{sec:relevant_modal}. We define $\mathsf{CL}$ as the axiom system comprising
\begin{enumerate}
\item $\mathsf{CPC}$ with (MP) and (US) where substitutions are functions from $Pr$ to $Fm_{\mathcal{L}}$;
\item for all axioms $\f$ of $\mathsf{L}$, an axiom $\Box_L \f$, and for all inference rules $\dfrac{\f_1 \ldots \f_n}{\ff}$ of $\mathsf{L}$, the rule $\dfrac{\Box_L \f \ldots \Box_L \f_n}{\Box_L \ff}$;
\item The Bridge Rule (BR) $\dfrac{\Box_L (\f \to \ff)}{\f \to \ff}$\, .
\end{enumerate}

\end{definition}

\begin{lemma}[Soundness]\label{lem:CL_sound}
For all $\mathsf{L}$ and all $\f$, if $\f \in Th(\mathsf{CL})$, then $\f \in \logic{CL}$.
\end{lemma}
\begin{proof}
Induction on the length of proofs. The base case is established by showing that all axioms of $\mathsf{CL}$ are valid in all $\logic{CL}$-frames. (i) All axioms of $\mathsf{CPC}$ are valid in all $W$-frames thanks to Prop.~\ref{prop:classconn}. (ii) The fact that $\Box_L \f$ is valid in all $\logic{CL}$-frames for each axiom $\f$ of $\mathsf{L}$ can be shown using Lemma \ref{lem:ded_BoxL} and, where applicable, using the C-variants of the $\logic{L}$ frame conditions. In most cases, this boils down to standard arguments \cite{Fuhrmann1988,Fuhrmann1990,RoutleyEtAl1982}.\footnote{We note in relation to axiom ($\Box$D) that Fuhrmann's frame condition $\forall s \exists x (Q s^{*}x \And Qs x^{*})$ corresponds to $\Box p \to \neg \Box \neg p$, not to $\Box \neg p \to \neg \Box p$, as stated in \cite{Fuhrmann1988,Fuhrmann1990}. The frame condition suffices to show that $\Box \neg p \to \neg \Box p$ is valid only if it is assumed that $x \leq x^{**}$.} The one case where the C-variant differs from the original frame condition is established as follows. Assume that the C-variant of the frame condition (X), namely (C-X) $\exists w (w \in W \And Q_L ws) \To s^{*} \leq s$, holds in a $W$-frame $\rs{F}$. To show that $\Box_L (p \lor \neg p)$ is valid in any $\rs{M}$ based on $\rs{F}$, take $w \in W$ and assume that $Q_L ws$. By (C-X), $s^{*} \leq s$. Thus, if $s \not\models p$, then $s^{*} \not\models p$ and so $s \models \neg p$.

To establish the induction step, we have to show that each rule of inference of $\mathsf{CL}$ preserves validity in $\logic{CL}$-frames. (i) (MP) preserves validity in $W$-frames thanks to Proposition \ref{prop:classconn}; and (US) clearly preserves validity in all frames.

(ii) The fact that the $\Box_L$-version of (Adj) preserves validity in $W$-frames is established using the satisfaction clause for $\Box_L$ and $\land$. The fact that $\Box_L$-versions of (Aff), (Con), ($\Box$-Mon) and ($\Box_L$-Mon) preserve validity in $W$-frames is established easily using Lemma \ref{lem:ded_BoxL} and satisfaction clauses for the operators involved. The fact that the $\Box_L$-version of (ER) preserves validity in $W$-frames satisfying (C-ER) $\forall s \exists wt (w \in W \And Q_L wt \And R s t s)$ is established as follows. Assume that $\Box_L \f$ is valid in some $\rs{M}$ based on $\rs{F}$ satisfying (C-ER). Take some $s$ and assume $s \models \f \to \ff$; we have to show that $s \models \ff$. Using (C-ER), there are $w \in W$ and $t \in S$ such that $Q_L w t$ and $Rsts$. Hence, $t \models \f$ and so $s \models \ff$. The fact that the $\Box_L$-version of (Nec) preserves validity in $W$-frames satisfying (C-Nec) $s \in Q_L(W) \And Q st \To t \in Q_L(W)$ is established as follows. Assume that $\Box_L \f$ is valid in some $\rs{M}$ based on $\rs{F}$ satisfying (C-Nec). Take some $w \in W$ and $s,t$ such that $Q_L ws$ and $Q st$. By (C-Nec), there is $u \in W$ such that $Q_L u t$. It follows from this that $t \models \f$. Hence, $w \models \Box_L \Box \f$.

(iii) The Bridge Rule preserves validity in $W$-frames by Lemma \ref{lem:ded_BoxL} and Proposition \ref{prop:classconn}.
\end{proof}

\begin{lemma}\label{lem:L-CL}
For all $\mathsf{L}$, $\vdash_{\mathsf{L}} \f$ iff $\vdash_{\mathsf{CL}} \Box_L \f$.
\end{lemma}
\begin{proof}
The fact that $\vdash_{\mathsf{L}} \f$ implies $\vdash_{\mathsf{CL}} \Box_L \f$ can be established by induction on the length of $\mathsf{L}$-proofs. The base case holds since $\Box_L \f$ is an axiom of $\mathsf{CL}$ for all axioms $\f$ of $\mathsf{L}$. The induction step is established using the fact that, by definition of $\mathsf{CL}$, if $\dfrac{\f_1 \ldots \f_n}{\ff}$ is an instance of an inference rule of $\mathsf{L}$, then $\dfrac{\Box_L \f \ldots \Box_L \f_n}{\Box_L \ff}$ is an instance of an inference rule of $\mathsf{CL}$.

Conversely, if $\not\vdash_{\mathsf{L}} \f$, then $\rs{M} \not\models \f$ for some $\logic{L}$-model $\rs{M}$. By Prop.~\ref{prop:+} there is a $\logic{CL}$-model $\rs{M'}$ such that $\rs{M'} \not\models \Box_L \f$ and so, by Lemma \ref{lem:CL_sound}, $\not\vdash_{\mathsf{CL}} \Box_L \f$.
\end{proof}

Lemma \ref{lem:L-CL} clarifies the role of the operator $\Box_L$ in our framework. We stress that $\Box_L$ is not an epistemic operator expressing attitudes of agents.\footnote{Nevertheless, formulas of the form $\Box_L (\f \to \ff)$ can be seen as expressing a form of information constraint; see Lemma \ref{lem:ded_BoxL}.} It follows from Lemma \ref{lem:L-CL} and the presence of the Bridge Rule that $ \vdash_{\mathsf{L}} \f \to \ff$ entails $ \vdash_{\mathsf{CL}} \f \to \ff$. In fact, we can prove a stronger claim.

\begin{proposition}
The following hold for all $\mathsf{L}$ such that the rule (Nec) is not an inference rule of $\mathsf{L}$:
\begin{enumerate}
    \item $ \vdash_{\mathsf{L}} \f$ entails $ \vdash_{\mathsf{CL}} \f$;
    \item $ \vdash_{\mathsf{CL}} \Box_L \f$ entails $ \vdash_{\mathsf{CL}} \f$.
\end{enumerate}
\end{proposition}
\begin{proof}
(i) The claim is established by induction on the length of $\mathsf{L}$-proofs. All implicational axioms of $\mathsf{L}$ are provable in $\mathsf{CL}$ by Lemma \ref{lem:L-CL} and the Bridge Rule; the axiom (X) $p \lor \neg p$ is provable in any $\mathsf{CL}$ since the propositional fragment of each $\mathsf{L}$ is included in $\mathsf{CPC}$ (this argument can be used also to show that the ``purely propositional'' axioms of $\mathsf{L}$ are provable in $\mathsf{CL}$). The cases of the induction step corresponding to rules with implicational conclusions are established using Lemma \ref{lem:L-CL} and (BR) as before (the arguments do not need to use the induction hypothesis). The case corresponding to (Adj) is established using the fact that $\vdash_{\mathsf{CPC}} \f \to (\ff \to (\f \land \ff))$ for all $\f, \ff$. We note that (Nec) is problematic: $\f$ is not necessarily an implication, so we can not use Lemma \ref{lem:L-CL} and (BR); using the induction hypothesis gives us only $\vdash_{\mathsf{CL}} \f$ and using only Lemma \ref{lem:L-CL} gives us $\vdash_{\mathsf{CL}} \Box_L \f$, from which we can infer only that $\vdash_{\mathsf{CL}} \Box_L \Box \f$ using the $\Box_L$-version of (Nec). (ii) follows from (i) and Lemma \ref{lem:L-CL}.
\end{proof}

Note that the converses of (i) and (ii) from the previous proposition do not hold.

\begin{proposition}[Relevant Reasoning]
Each $\mathsf{CL}$ is closed under the Relevant reasoning meta-rule  \eqref{eq:RR}
\begin{equation*}
\dfrac{\vdash_{\mathsf{L}} \f_1 \land \ldots \land \f_n \to \ff}{\vdash_{\mathsf{CL}} \Box \f_1 \land \ldots \land \Box \f_n \to \Box \ff}
\end{equation*}
for all $n > 0$.
\end{proposition}
\begin{proof}
 If $\vdash_{\mathsf{L}} \bigwedge_{i \leq n} \f_i  \to  \ff $, then $\vdash_{\mathsf{L}} \bigwedge_{i \leq n} \Box \f_i  \to  \Box \ff $ using monotonicity and regularity of $\Box$ in $\mathsf{L}$, and so $\vdash_{\mathsf{CL}} \Box_L \big (\bigwedge_{i \leq n} \Box \f_i  \to  \Box \ff \big )$ by Lemma \ref{lem:L-CL}. But then $\vdash_{\mathsf{CL}} \bigwedge_{i \leq n} \Box \f_i  \to  \Box \ff$ follows using the Bridge Rule.
\end{proof}

Let $\mathsf{(C)L} \in \{ \mathsf{L}, \mathsf{CL} \}$. A \emph{$\mathsf{(C)L}$-theory} is any set $\Gamma$ of formulas such that (i) $\f \in \Gamma$ and $\ff \in \Gamma$ only if $\f \land \ff \in \Gamma$; and (ii) if $\f \in \Gamma$ and $\vdash_{\mathsf{(C)L}} \f \to \ff$, then $\ff \in \Gamma$. A $\mathsf{(C)L}$-theory $\Gamma$ is \emph{prime} iff (iii) $\f \lor \ff \in \Gamma$ only if $\f \in \Gamma$ or $\ff \in \Gamma$; it is \emph{proper} iff $\Gamma \neq Fm$. 

A pair of sets of formulas $(\Gamma, \Delta)$ is \emph{$\mathsf{(C)L}$-independent} iff there are no finite non-empty sets $\Gamma' \subseteq \Gamma$ and $\Delta' \subseteq \Delta$ such that $$\vdash_{\mathsf{(C)L}} \bigwedge \Gamma' \to \bigvee \Delta' \, .$$ 

\begin{lemma}[Extension Lemma]\label{lem:extension}
For all $\mathsf{L}$:
\begin{enumerate}
    \item If $(\Gamma, \Delta)$ is $\mathsf{L}$-independent, then there is a prime $\mathsf{L}$-theory $\Sigma$ such that $\Gamma \subseteq \Sigma$ and $\Delta \cap \Sigma = \emptyset$.
    \item If $(\Gamma, \Delta)$ is $\mathsf{CL}$-independent and both $\Gamma$ and $\Delta$ are non-empty, then there is a non-empty proper prime $\mathsf{CL}$-theory $\Sigma$ such that $\Gamma \subseteq \Sigma$ and $\Delta \cap \Sigma = \emptyset$.
\end{enumerate}
\end{lemma}
\begin{proof}
(i) If $\Gamma = \emptyset$, then let $\Sigma := \emptyset$. If $\Gamma \neq \emptyset$ and $\Delta = \emptyset$, then let $\Sigma := Fm$. If both $\Gamma$ and $\Delta$ are non-empty, then use the standard ``prime extension'' argument (cf.~Theorem 5.17 in \cite{Restall2000}, for example).

(ii) This is established similarly as the first claim (note that $\Sigma$ needs to be proper if $\Delta \neq \emptyset$). In fact, this is the well-known Lindenbaum Lemma.
\end{proof}

A $\mathsf{(C)L}$-theory $\Gamma$ is \emph{maximal $\mathsf{(C)L}$-consistent} iff it is a proper $\mathsf{(C)L}$-theory such that each $\mathsf{(C)L}$-theory $\Delta \supsetneq \Gamma$ is not proper. It is easily shown that a $\mathsf{CL}$-theory is non-empty, proper and prime iff it is maximal $\mathsf{CL}$-consistent. It is also easily shown that if $\Gamma$ is a non-empty proper prime $\mathsf{CL}$-theory, then $Th(\mathsf{CL}) \subseteq \Gamma$.

\begin{definition}
The \emph{canonical $\ax{CL}$-frame} is the structure $$\rs{F}^{\ax{CL}} = (S^{\ax{CL}}, \leq^{\ax{CL}}, W^{\ax{CL}}, R^{\ax{CL}}, *^{\ax{CL}}, Q^{\ax{CL}}, Q_L^{\ax{CL}})$$ where
 \begin{itemize}
 \item $S^{\ax{CL}}$ is the set of all prime $\ax{L}$-theories;
 \item $\leq^{\ax{CL}}$ is set inclusion;
 \item $W^{\ax{CL}}$ is the set of all non-empty proper prime $\ax{CL}$-theories;
 \item $R^{\ax{CL}} stu$ iff $\forall \f, \ff$, if $\f \to \ff \in s$ and $\f \in t$, then $\ff \in u$;
 \item $s^{*^{\ax{CL}}} = \{ \f \mid \neg\f \not\in s \}$;
 \item $Q^{\ax{CL}}st$ iff, $\forall\f$, $\Box \f \in s$ only if $\f \in t$.
  \item $Q^{\ax{CL}}_L st$ iff, $\forall \f$, $\Box_L \f \in s$ only if $\f \in t$;
 \end{itemize}
 The \emph{canonical $\ax{CL}$-model} is $\rs{M}^{\ax{CL}} = (\rs{F}^{\ax{CL}}, V^{\ax{CL}})$ where $V^{\ax{CL}} : Pr \to 2^{S^{\ax{CL}}}$ such that $V^{\ax{CL}}(p) = \{ s \in S^{\ax{CL}} \mid p \in s \}$.
\end{definition}

Canonical models are well defined since, for each $\mathsf{L}$, any prime $\mathsf{CL}$-theory is a prime $\mathsf{L}$-theory (note that if $\vdash_{\mathsf{L}} \f \to \ff$, then $\vdash_{\mathsf{CL}} \Box_L (\f \to \ff)$ by Lemma \ref{lem:L-CL} and then $\vdash_{\mathsf{CL}} \f \to \ff$ using the Bridge Rule).

Note that $\rs{M}^{\mathsf{CL}}$ is defined almost exactly as the canonical $\mathsf{L}$-model $\rs{M}^{\mathsf{L}}$ (cf.~\cite{Fuhrmann1988,Fuhrmann1990}); the only difference is that in $\rs{M}^{\mathsf{L}}$ we have $L^{\mathsf{L}} = \{ s \mid Th(\mathsf{L}) \subseteq s \}$ instead of $W^{\mathsf{CL}}$. It follows that $\rs{M}^{\mathsf{CL}}$ automatically satisfies all the $L$-frame conditions $\Phi$ that are identical to their C-variant and are satisfied by $\rs{M}^{\mathsf{L}}$.

\begin{lemma}\label{lem:canonical}
For all $\mathsf{L}$, the structure $\rs{M}^{\ax{CL}}$ is a $\logic{CL}$-model.
\end{lemma}
\begin{proof}
First we have to show that $\rs{M}^{\mathsf{CL}}$ is a $W$-model, that is, (i) relations $R^{\ax{CL}}, *^{\ax{CL}}, Q^{\ax{CL}}$ and $Q_L^{\ax{CL}}$ satisfy the required tonicity conditions, (ii) $\rs{M}^{\mathsf{CL}}$ is bounded, (iii) $W$ is a set of possible worlds and (iv) conditions \eqref{eq:Q_L-s<t}-\eqref{eq:Q_L-ss} are satisfied. Then we have to show that (v) $\rs{M}^{\ax{CL}}$ satisfies the $\logic{CL}$ frame conditions. In the remainder of the proof, we will mostly omit the superscript $\mathsf{CL}$.

Claim (i) follows easily from the definitions. Claim (ii) is established as follows. Let $0 := \emptyset$ and $1 := Fm$; both are prime $\mathsf{L}$-theories. Then $0$ is the least element (with respect to set inclusion) of $S$ and $1$ is the greatest element. Next, $1^{*} = \{ \f \mid \neg \f \not\in 1 \} = \emptyset = 0$ and $0^{*} = \{ \f \mid \neg\f \not\in 0 \} = Fm = 1$. Next, clearly $Q_{(L)} 00$ for all $s$. Next, if $Q_{(L)} 1 s$ and $s \neq 1$, then there is $\f \notin s$ and so $\Box_{(L)} \f \notin 1$, which is a contradiction; hence, $s = 1$. Next, $R010$ by the definition of $R$. Finally, if $R 1st$, $s \neq 0$ and $t \neq 1$, then there is $\f \in s$ and $\ff \notin t$ such that $\f \to \ff \in 1$, which contradicts the definition of $R$.

(iii) We know that non-empty proper prime $\mathsf{CL}$-theories are maximal $\mathsf{CL}$-consistent theories, i.e.~$\f \in s$ iff $\neg \f \notin s$. Hence $w^{*} = \{ \f \mid \neg\f \notin w \} = w$. The fact that $Rwww$ follows from the fact that $\vdash_{\mathsf{CPC}} (\f \land (\f \to \ff)) \to \ff$. Now assume that $Rwst$ and $s \neq 0$; we have to prove that $w \subseteq t$. Thus assume $\f \in w$ and $\ff \in s$. Since $\vdash_\mathsf{CPC} \f \to (\ff \to \f)$, we have $\ff \to \f \in w$, and so $\f \in w$ by the definition of $R$. Since $\f$ is arbitrary, we established $w \subseteq t$. Finally, assume $Rwst$ and $t \neq 1$; we have to prove that $s \subseteq w^{*}$. Hence, assume that $\neg \f \in w$, $\ff \notin t$ and, towards a contradiction, that $\f \in s$. Since $\vdash_{\mathsf{CPC}} \neg \f \to (\f \to \ff)$, we have $\f \to \ff \in w$ and so $\ff \in t$ by the definition of $R$. This is a contradiction; hence, if $\f \in s$, then $\neg\f \not\in w$, meaning in general that $s \subseteq w^{*}$.\footnote{Note that $Rwst \To w \subseteq t$ does not hold since $Rw0t$ for all $w$ and $t$. Similarly, $Rwst \To s \subseteq w^{*}$ does not hold since $Rws1$ for all $w$ and $s$.}\label{ft:01canoncial}

(iv) Note that \eqref{eq:Q_L-s<t} follows from the fact that $\Box_L (\f \to \f)$ is a theorem of $\mathsf{CL}$. Now we prove that the canonical frame satisfies \eqref{eq:Q_L-ss}. Take any $s \in S$; we will prove that there are $w \in W$ and $t \in S$ such that $Q_L wt$ and $Rtss$. If $s = 1$, then $Q_L ws$ for all $w \in W$ and $Rsss$ and so we are done. If $s = 0$, then $Q_L w 1$ for all $w \in W$ and $R 1 ss$ and we are done. Finally, assume that $s \neq 1$ and $s \neq 0$. First, it is easily shown that the pair $\big( Th(\ax{L}), \{ \f \to \ff \mid \f \in s \And \ff \notin s \}\big)$ is $\ax{L}$-independent and so, by the Extension Lemma \ref{lem:extension}, there is a non-empty prime $\ax{L}$-theory $t$ such that $Rtss$.
 If $t = 1$, then also $Q_L w t$ for all $w \in W$ and we are done. If $t \neq 1$, then we can show that the pair
$$\big( Th(\ax{CL}), \{ \Box_L \f \mid \f \notin t \} \big)$$ is $\ax{CL}$-independent. If it were not, then
\begin{itemize}
\item $\vdash_{\ax{CL}} \bigvee_{i < n} \Box_L \f_i$ for some $n > 0$, and so
 \item $\vdash_{\ax{CL}} \Box_L \bigvee_{i < n} \f_i$ by the properties of $\to$ in $\mathsf{CL}$ (namely, $\vdash_{\mathsf{CL}} \Box_L \f \lor \Box_L \ff \to \Box_L (\f \lor \ff)$), hence
 \item $\vdash_{\ax{L}} \bigvee_{i < n} \f_i$ by Lemma \ref{lem:L-CL}, so
 \item $\bigvee_{i < n} \f_i \in t$ by the construction of $u$, and so
 \item $\f_i \in t$ for some $i < n$ since $u$ is prime.
 \end{itemize}
 This is a contradiction, so the pair is $\ax{CL}$-independent. Note also that both sets in the pair are non-empty (since, recall, $t \neq 1$). It follows using the Extension Lemma \ref{lem:extension} that there is a non-empty maximal consistent $\ax{CL}$-theory $w$ such that $Q_L wt$. 

(v) Assume first that $\Phi$ is a $L$-frame condition that is identical to its C-variant. As noted above, $\rs{F}^{\mathsf{CL}}$ has $\Phi$ iff $\rs{F}^{\mathsf{L}}$ has $\Phi$. However, the fact that that $\rs{F}^{\mathsf{L}}$ has $\Phi$ follows from the well-know result that all $\mathsf{L}$ considered in this paper are canonical \cite{Fuhrmann1988,Fuhrmann1990,RoutleyEtAl1982}. We will omit the details.

Assume now that $\Phi$ is a $L$-frame condition such that the C-variant $\Phi'$ is not identical to $\Phi$. We reason by cases. Take (C-X) $s \in Q_L(W) \To s^{*} \leq s$. If $\mathsf{L}$ contains axiom $p \lor \neg p$, then $\mathsf{CL}$ contains axiom $\Box_L (p \lor \neg p)$. Hence, if $Q_L ws$ for some $w \in W$, then $\f \lor \neg \f \in s$ for all $\f$. Now assume that $\ff \in s^{*}$. Hence, $\neg \ff \not\in s$ and so $\ff \in s$. In general, $s^{*} \subseteq s$ as we wanted to show.  

Now take (C-ER), $\forall s \exists w, t ( w \in W \And Q_L wt \And  Rsts)$. Fix any $s \in S$. We first construct a prime $\mathsf{L}$-theory $t$ such that $Rsts$. In order to do this, we show that the pair $$ \big ( Th(\mathsf{L}), \{ \f \mid \exists \ff (\f \to \ff \in s \And \ff \not\in s) \} \big)$$ is $\mathsf{L}$-independent if  (ER) is an admissible rule of $\mathsf{L}$. If the pair were not $\mathsf{L}$-independent, then
\begin{itemize}
\item $\vdash_{\mathsf{L}} \bigvee_{i < n} \f_i$ for some $n > 0$, which entails that
\item $\vdash_{\mathsf{L}} (\bigvee_{i < n} \f_i \to \bigvee_{i < n} \ff_i) \to \bigvee_{i < n} \ff_i$ using (ER) where, for each $i < n$, $\ff_i$ is a formula such that $\f_i \to \ff_i \in s$ and $\ff_i \not\in s$; but this entails that
\item $\bigvee_{i < n} \ff_i \in s$ since $\bigwedge_{i < n} (\f_i \to \ff_i) \in s$ and so $(\bigvee_{i < n} \f_i \to \bigvee_{i < n} \ff_i) \in s$. But this means that
\item $\ff_i \in s$ for some $i < n$.
\end{itemize}
This is a contradiction and so the pair is $\mathsf{L}$-independent. Consequently, there is a prime $\mathsf{L}$-theory $t$ such that $Th(\mathsf{L}) \subseteq t$ and $Rsts$, using the Extension Lemma \ref{lem:extension}. As a second step of the argument, we show that there is a non-empty proper prime $\mathsf{CL}$-theory $w$ such that $Q_L wt$. If $t = 1$, then $Q_L wt$ for all $w \in W$ and we are done. If $t \neq 1$, then we can show that the pair
$$ \big ( Th(\mathsf{CL}), \{ \Box_L \f \mid \f \notin t \} \big)$$ is $\mathsf{CL}$-independent. (The argument is similar to the one establishing \eqref{eq:Q_L-ss}.)
 Note also that both sets in the pair are non-empty (since it is assumed that $t \neq 1$) and so, using the Extension Lemma \ref{lem:extension}, there is a non-empty proper prime $\mathsf{CL}$-theory $w$ such that $Q_L wt$. (C-Nec) is dealt with in a similar fashion.

\end{proof}

\begin{lemma}[Truth Lemma]\label{lem:truth}
For all $\mathsf{L}$ and all $\f$, $(\rs{M}^{\mathsf{CL}}, s) \models \f$ iff $\f \in s$.
\end{lemma}
\begin{proof}
Induction on the complexity of $\f$. The base holds by definition and the cases of the induction step where the main connective is propositional are established as usual in relevant logic \cite{Restall2000,RoutleyEtAl1982}. The cases where the main connective is $\Box$ or $\Box_L$ are established as usual in relevant modal logic \cite{Fuhrmann1988,Fuhrmann1990} using monotonicity and regularity of the box operators (in $\mathsf{L}$ since the claim of the lemma concerns an arbitrary prime $\mathsf{L}$-theory). 
\end{proof}

\begin{theorem}
For all $\mathsf{L}$, $Th(\mathsf{CL}) = \logic{CL}$.
\end{theorem}
\begin{proof}
Soundness is established as Lemma \ref{lem:CL_sound}. Completeness follows from Lemmas \ref{lem:extension}, \ref{lem:canonical} and \ref{lem:truth}.
\end{proof}

\section{Conclusion}

In this paper we studied a framework for epistemic logic that avoids the logical omniscience problem by introducing non-standard states, but where the epistemic attitudes of agents are regimented by a relevant logic. Unlike the classic non-standard-states approaches of Levesque \cite{Levesque1984} and Lakemeyer \cite{Lakemeyer1987}, the relevant logic regimenting attitudes is not the $\land,\lor,\neg$-fragment of $\logic{E}$, but any relevant modal logic from a wide variety of logics previously studied by Fuhrmann \cite{Fuhrmann1988,Fuhrmann1990}, for example. Hence, our approach provides a more realistic relevant formalization of attitudes towards implicational statements, which is an improvement with respect to  the approach of Fagin et al.~\cite{FaginEtAl1995a} as well. In comparison to earlier work on relevant epistemic logic \cite{BilkovaEtAl2016,BilkovaEtAl2010,Sedlar2015,Sedlar2016a}, our framework combines relevant modal logic with classical propositional logic and it uses natural Hilbert-style proof theory without unnecessary linguistic assumptions. The main technical result of the paper is a modular completeness theorem, parametrized by the relevant modal logic governing the agents' reasoning. 

There is a number of issues to pursue in the future. Firstly, it would be interesting to consider a version of our framework without the assumption of conjunctive regularity $\Box \f \land \Box \ff \to \Box (\f \land \ff)$. Such a framework needs to be based on neighborhood semantics. Secondly, it would be interesting to look at extensions of our logics with operators representing group epistemic notions such as common and distributed belief, or dynamic phenomena such as public announcements. A combination of our framework with the approach of \cite{PuncocharSedlar2021a} is an option to consider. Thirdly, it is tempting to consider a first-order version of the present framework. 

\medskip
\paragraph{Acknowledgement.} Igor Sedlár acknowledges the support of the long-term strategic development financing of the Institute of Computer Science (RVO:67985807) and the support of the Czech Science Foundation
project GA22-01137S. We thank Nicholas Ferenz for valuable comments on an earlier version of the paper. We are grateful to three anonymous reviewers for AiML for a number of useful suggestions. The paper was presented at the workshop \emph{Relevant Logic Today} organized by Peter Verdée at UCLouvain in May 2022. We are grateful to the audience at the workshop for helpful discussion.
 

\appendix
\section{Technical appendix}
This technical appendix contains the proof of Proposition \ref{prop:+}.

\medskip

\noindent
\textbf{Proposition \ref{prop:+}}. 
\textit{
For each $L$-model $\rs{M}$ for $\logic{L}$ with a set of states $S$ there is a $W$-model $\rs{M'}$ for $\logic{CL}$ with the set of states $S' \supsetneq S$ such that, for all $\f \in Fm_{\mathcal{L}}$:
\begin{enumerate}
\item for all $s \in S$, $(\rs{M}, s) \models \f$ iff $(\rs{M'}, s) \models \f$;
\item if $\rs{M} \not\models \f$, then $\rs{M'} \not\models \Box_L \f$.
\end{enumerate}
}

\noindent
\begin{proof}
Let $\relstruct{M} = (S, \leq, L , R, *, Q, Q_L, V)$ be a $L$-model. The structure $\relstruct{M^{+}} = (S^{+}, \leq^{+}, W^{+}, R^{+}, *^{+}, Q^{+}, Q^{+}_L, V^{+})$ is defined as follows:
\begin{align*}
    &S^{+} = S \cup \{w, 0, 1 \} \\
    &\mathord{\leq^{+}} = \mathord{\leq} \cup \{(w,w)\} \cup \{ (s, 1) \mid s \in S^{+} \} \cup \{ (0, s) \mid s \in S^{+} \} \\
    &W = \{w\}\\
    &R^{+} = R \cup \{(w,w,w)\} \cup  \{ (0, s, t), (s, 0, t), (s, t, 1) \mid s,t \in S^{+} \} \\
    &*^{+} = * \cup \{(w,w)\} \cup \{ (0, 1), (1, 0) \} \\
    &Q^{+} = Q \cup \{(w,w)\} \cup \{ (s, 1) \mid s \in S^{+} \} \cup \{ (0, s) \mid s \in S^{+} \} \\
    &Q_L^{+} = Q_L \cup \{(w,w)\} \cup \{(w,s) \mid s \in L \} \cup \{ (s, 1) \mid s \in S^{+} \} \cup \{ (0, s) \mid s \in S^{+} \} \\
    &V^{+} (p) = V(p) \cup \{ 1 \} \text{ for all } p
\end{align*}
We prove first that for all $\rs{M}$, the structure $\rs{M^{+}}$ is a $W$-model (Claim \ref{claim:+1}), then we prove that for all $s \in S$ and all $\f \in Fm_{\mathcal{L}}$, $(\rs{M}, s) \models \f$ iff $(\rs{M^{+}}, s) \models \ff$ (Claim \ref{claim:+2}). It follows from the construction of $\rs{M^{+}}$ and Claim \ref{claim:+2} that if $\rs{M} \not\models \f$, then $\rs{M^{+}} \not\models \Box_L \f$ (since $(\rs{M}, s) \not\models \f$ for some $s \in L$ implies that $(\rs{M^{+}}, w) \not\models \Box_L \f$). Finally, we prove that if $\rs{M}$ is an $\logic{L}$-model, then $\rs{M^{+}}$ is a $\logic{CL}$-model (Claim \ref{claim:+3}).

\begin{claim}\label{claim:+1}
For all $\rs{M}$, the structure $\rs{M^{+}}$ is a $W$-model.
\end{claim}
\noindent Firstly, we have to show that $\rs{M^{+}}$ is a model, i.e.~$\leq^{+}$ is a partial order and $R^{+}$, $\ast^{+}$, $Q^{+}$ and $Q_L^{+}$ satisfy the required tonicity conditions. A simple analysis of cases shows that $\leq^{+}$ is a partial order. To show that $R^{+} \in S^{+}(\downarrow\downarrow\uparrow)$, assume for instance that $R^{+}stu$ and $v \leq s$. If $\{ s,t,u \} \subseteq S$, then either $v \in S$ or $v = 0$; in both cases $R^{+}vtu$. If $s \not\in S$ and $s = 0$, then $v = 0$ and we are done since $R^{+}0tu$ for all $t,u$. If $s = w$, then either $s = t = u$ and so either $v = w$ or $v = 0$ and we are done, or $t = 0$ and we are done, or $u = 1$ and we are done. Finally, if $s = 1$, then either $t = 0$ or $u = 1$, but then we are done in both cases. The cases $t \not\in S$ and $u \notin S$ are dealt with similarly. The fact that the other relations satisfy the required tonicity conditions is established by a similar tedious examination of cases.

Secondly, we have to show that $\rs{M^{+}}$ is a bounded model and that $w$ is a possible world. This is easily checked. Thirdly, we have to show that the frame conditions \eqref{eq:Q_L-s<t} and \eqref{eq:Q_L-ss} are satisfied. Note that $Q^{+}_L wu$ iff $u \in L \cup \{ w, 1 \}$. We deal with \eqref{eq:Q_L-s<t} first. If $u \in L$ and $R^{+}ust$, then either $s, t \in S$ and we are done thanks to \eqref{eq:L1}, or $s = 0$ or $t = 1$, and then we are done thanks to the properties of $\leq^{+}$. If $u = w$, then by \eqref{eq:world3} either $s = 0$ in which case we are done or $w \leq^{+} t$. If $w = t$, then either $s = w$ or $s = 0$ and we are done; if $t = 1$, then we are done no matter what $s$ is. Now we check \eqref{eq:Q_L-ss}. If $s \in S$, then there is $u \in L$ such that $Russ$ by \eqref{eq:world4}; hence $R^{+} uss$, but we also know that $Q^{+}_L w u$ and so we are done. If $s \notin S$, then $R^{+}wss$ by easy inspection of cases, but we also know that $Q^{+}_L ww$, and so we are done. This concludes the proof of Claim \ref{claim:+1}.   

\begin{claim}\label{claim:+2}
For all $s \in S$ and all $\f \in Fm_{\mathcal{L}}$, $(\rs{M}, s) \models \f$ iff $(\rs{M^{+}}, s) \models \f$.
\end{claim}
\noindent The proof is by  induction on the complexity of $\vp$. The base case and the induction steps for $\land, \lor$ are trivial. 
$(\relstruct{M}, s) \models \neg\f$ iff $(\relstruct{M}, s^{*}) \not\models \f$ iff $(\relstruct{M^+}, s^{*}) \not\models \f$ (by IH) iff $(\relstruct{M^+}, s^{*^+}) \not\models \f$ (by the fact, easily confirmed by inspection of the definition, that if $s \in S$, then $s^{*} = s^{*^+}$) iff $(\relstruct{M^+}, s) \models \neg \f$. If $(\rs{M}, s) \not\models \f \to \ff$, then there are $t,u \in S$ such that $Rstu$ and $(\rs{M^{+}}, t) \models \f$ and $(\rs{M^{+}}, u) \not\models \ff$ by IH. The rest follows from $R \subseteq R^{+}$. Conversely, if $(\rs{M^{+}}, s) \not\models \f \to \ff$, then there are $t, u \in S^{+}$ such that $R^{+}stu$, $(\rs{M^{+}}, t) \models \f$ and $(\rs{M^{+}}, u) \not\models \ff$. By inspection of the definition of $R^{+}$, either $t,u \in S$, or $t = 0$, or $u = 1$. Since we already know that $\rs{M^{+}}$ is a bounded model, the latter two options are ruled out by Lemma \ref{lem:0empty1full}. Hence, $t, u \in S$ but in this case $Rstu$ and so $(\rs{M}, s) \not\models \f \to \ff$ by IH. If $(\rs{M}, s) \not\models \Box_L \f$, then $(\rs{M^{+}}, s) \not\models \Box_L \f$ by $Q_L \subseteq Q^{+}_L$ and IH. Conversely, if $(\rs{M^{+}}, s) \not\models \Box_L \f$, then there is $t$ such that $Q^{+}_L st$ and $(\rs{M^{+}}, t) \not\models \f$. By inspection of the definition of $Q^{+}_L$ we see that either $t = 1$, which contradicts Lemma \ref{lem:0empty1full}, or $t \in S$, in which case $Q_L st$ and so we are done by IH. The case of $\Box$ is analogous. This concludes the proof of Claim \ref{claim:+2}.  

\begin{claim}\label{claim:+3}
if $\rs{M}$ is an $\logic{L}$-model, then $\rs{M^{+}}$ is a $\logic{CL}$-model.
\end{claim}
\noindent The proofs for (B), (CB), (W), (C), (WB) and (M) follow the same strategy, hence we show the details for (B) only. Assume that $R^{+}stx$ and $R^{+}xuv$; we have to prove that there is $y$ such that $R^{+} tu y$ and $R^{+} s y v$. Let $T = \{ s, t, u, v, x \}$. First, if $T \subseteq S$, then $Rstx \And Rxuv$ and so we are done, since (B) holds in $\rs{M}$ and $R \subseteq R^{+}$. Second, if $1 \in T$ or $0 \in T$, then we distinguish three cases:
\begin{enumerate}
\item If $0 \in \{ s, t, u \}$ or $v = 1$, then we are done. (For instance, if $s = 0$, then $R^{+} s 1 v$ and $R^{+} t u 1$; the other cases are similar.)

\item If $x = 0$, then by \eqref{eq:R1st} either $s = 0$ or $t = 0$ and we are in (i). If $x=1$, then by \eqref{eq:R1st} either $u = 0$ or $v = 1$; in both cases we are in (i).  We will use \eqref{eq:R1st} without explicit reference below.

\item If $s=1$, then $t = 0$ (i) or $x = 1$ (ii). If $t = 1$, then $s = 0$ (i) or $x = 1$ (ii). If $u = 1$, then $x = 0$ or $v = 1$ (ii). If $v = 0$, then $x = 0$ (ii) or $u = 0$ (i).
\end{enumerate}
Third, if $T \subseteq S \cup \{ w \}$, then we are either in case (i) or $T = \{ w \}$. In the latter case, set $y = w$ and we are done. These three groups of cases exhaust all possibilities and so $\rs{M^{+}}$ has to satisfy (B) if $\rs{M}$ does.

(DN) and (Rd) are preserved since they hold for $s \in \{ w, 0, 1 \}$ irrespectively of the properties of $\rs{M}$. (Cp) Assume that $R^{+} stu$. If $T = \{ s, t, u \} \subseteq S$, then $\rs{M^{+}}$ satisfies the fame condition if $\rs{M}$ does. If $s = 0$ or $t = 0$ or $u = 1$ then $R^{+} s u^{*^{+}} t^{*^{+}}$ holds by definition of $*^{+}$ and $R^{+}$. The cases where $s = 1$ or $t = 1$ or $u = 0$ reduce to the previous cases. If $T \subseteq S \cup \{ w \}$, then either $T \subseteq S$ or $T = \{ w \}$; we are done in both cases. To prove that if $\rs{M}$ satisfies (X), then $\rs{M^{+}}$ satisfies (C-X) $s \in Q_L(W) \To s^{*} \leq s$, assume $s \in Q^+_L(w)$, i.e. $s=w$ , $s=1$ or $s \in L$. If $s=1$, $0 \leq^+ 1$; if $s=w$ $w \leq^+ w$; and if $s \in L$ then $s^{*} \leq s$ by assumption and so $s^{\ast^+} \leq^+ s$.

To show that $\rs{M^{+}}$ satisfies (C-ER) $\forall s \exists x (x \in Q^{+}_L(w) \And R^{+} s x s)$ if $\rs{M}$ satisfies (ER), we reason as follows. If $s \in \{ w, 0, 1 \}$, then $R^{+} s w s$, and we know that $Q^{+}_L ww$. If $s \in S$, then by (ER) there is $x \in L$ such that $Rsxs$, and so $Q^{+}_L wx$ and $R^{+}s x s$.

To show that $\rs{M^{+}}$ satisfies (C-Nec) $\forall s \forall x (Q^{+}_L wx \And Q^{+} xs) \To Q^{+}_L ws$ if $\rs{M}$ satisfies (Nec), we reason as follows. Assume that $Q^{+}_L wx \And Q^{+} xs$. WE have to prove $Q^{+}_L ws$. If $T = \{ x, s \} \subseteq S$, then $x \in L$ and $Qxs$. Using (Nec), we obtain $s \in L$, which entails that $Q^{+}_L ws$. If $0 \in T$ or $1 \in T$, then we reason as follows. If $x = 0$, then $Q^{+} wx$ entails that $w = 0$, which is a contradiction. If $x = 1$, then $Q^{+} xs$ entails $s = 1$ and then $Q^{+}_L ws$ by definition of $Q^{+}_L$, and we are done. If $s = 0$, then $Q^{+}xs$ entails $x = 0$ which we already know to lead to a contradiction. The only remaining possibility is that $T \subseteq S \cup \{ w \}$. We have already checked the case $T \subseteq S$. If, on the other hand, $w \in T$, then we reason as follows. If $x = w$, then $Q^{+} xs$ entails that either $s = w$ or $s = 1$. In both cases $Q^{+}_L ws$. If $s = w$, then $Q^{+}_L ws$ as before. This exhausts all possibilities and so we are done.

To show that $\rs{M^{+}}$ satisfies ($\Box$K) if $\rs{M}$ does, assume that $R^{+} stx$ and $Q^{+} xu$; we have to prove that there are $y,z$ such that $Q^{+}ty$, $Q^{+}sz$ and $R^{+}zyu$. First, if $T = \{ s, t, u, x \} \subseteq S$, then $Rstx \And Qxu$, and so we are done using the assumption that $\rs{M}$ satisfies ($\Box$K) and $R \subseteq R^{+}$, $Q \subseteq Q^{+}$. Second, if $0 \in T$ or $1 \in T$, then we reason as follows. (i) If $s = 0$ or $t = 0$ or $u = 1$, then we can easily find $y,z \in \{ 0, 1 \}$ such that $Q^{+}ty$, $Q^{+}sz$ and $R^{+}zyu$; (ii) if $x = 0$, then $s = 0$ or $t = 0$, which reduce to case (i), and if $x = 1$, then $u = 1$, which also reduces to case (i); (iii) if $s = 1$, then $t = 0$ (i) or $x = 1$ (ii); if $t = 1$, then $s = 0$ (i) or $x = 1$ (ii); and if $u = 0$, then $x = 0$ (ii). Third, if $T \subseteq S \cup \{ w \}$, then either $T \subseteq S$, which is the first case, or $w \in T$; but then it can be shown that $T = \{ w \}$ from which it follows easily that $Q^{+}tw$, $Q^{+}sw$ and $R^{+}wwu$. This exhausts all possibilities and so we are done.

If $\rs{M}$ satisfies ($\Box$T), then so does $\rs{M^{+}}$ since $Q^{+}ss$ if $s \in \{ w, 0, 1 \}$. ($\Box$D) is dealt with similarly; if $s \in \{ w, 0, 1 \}$ then a suitable $x \in \{ w, 0, 1 \}$ is easily found. To show that ($\Box$4) is satisfied in $\rs{M^{+}}$ if it is satisfied in $\rs{M}$, we reason as follows. Assume that $Q^{+}st$ and $Q^{+}tu$; we have to show that $Q^{+}su$. First, if $T = \{ s, t, u \} \subseteq S$, then $Q^{+}su$ follows from the assumption that ($\Box$4) holds in $\rs{M}$ and $Q \subseteq Q^{+}$. Second, if $0 \in T$ or $1 \in T$, then we reason by cases as follows: (i) if $s = 0$ or $u = 1$, then trivially $Q^{+}su$. (ii) If $s = 1$, then $t = 1$ and so $u = 1$, which brings us back to case (i); if $u = 0$, then $t = 0$ and so $s = 0$, which brings us back to (i). Third, if $T \subseteq S \cup \{ w \}$, then either $T \subseteq S$, in which case we are done, or $T = \{ w \}$, in which case of course $Q^{+}su$. This exhausts all possibilities and so we are done. Preservation of ($\Box$5) is established in a similar way. This concludes the proof of Claim \ref{claim:+3} and Proposition \ref{prop:+}.
\end{proof}

\end{document}